\documentclass[letterpaper,12pt]{amsart}
\title{Identifiability of Large Phylogenetic Mixture Models}

\author{John A.~Rhodes}
\address{Department of Mathematics and Statistics \\
University of Alaska, Fairbanks AK 99775}
\email{j.rhodes@alaska.edu}

\author{Seth Sullivant}
\address{Department of Mathematics \\
North Carolina State University, Raleigh, NC 27695}
\email{smsulli2@ncsu.edu}

\date{}

\usepackage{latexsym,array,delarray,amsthm,amssymb,epsfig}

\theoremstyle{plain}
\newtheorem{thm}{Theorem}[section]
\newtheorem{lemma}[thm]{Lemma}
\newtheorem{prop}[thm]{Proposition}

\theoremstyle{definition}
\newtheorem{defn}[thm]{Definition}

\theoremstyle{remark}

\newcommand{\rr}{\mathbb{R}}

\newcommand{\bfa}{\mathbf{a}}

\newcommand{\bfm}{\mathbf{m}}

\newcommand{\bfT}{\mathbf{T}}

\newcommand{\calm}{\mathcal{M}}

\renewcommand{\flat}{\operatorname{Flat}}

\newcommand{\fraks}{{\frak S}}
\newcommand{\rank}{\mathrm{rank}}
\newcommand{\nullity}{\operatorname{nullity}}
\newcommand{\im}{\operatorname{im}}
\newcommand{\adj}{\operatorname{adj}}

\newcommand{\inD}[1][\relax]{\def\argone{#1}\def\temprelax{\relax}
  \ifx\argone\temprelax\right.\else\,\middle|#1\right.{}\fi}

\begin{document}
\maketitle

\begin{abstract}
Phylogenetic mixture models are statistical models of character evolution allowing for heterogeneity. Each of the classes
in some unknown partition of the characters may evolve by different processes, or even along different trees. 
The fundamental question of whether parameters of such a model are identifiable is difficult to address, due to the complexity of the parameterization. We analyze mixture models on large trees, with many mixture components, showing that both numerical and tree parameters are indeed identifiable in these models when all trees are the same.  We also explore the extent to which our algebraic techniques can be employed to extend the result to mixtures on different trees. 
\end{abstract}


\section{Introduction}

A fundamental question about any parametric statistical model is whether or not the parameters of that model are \emph{identifiable}; that is, does a probability distribution arising from the model uniquely determine the parameters that produced it.  Establishing the identifiability of parameters is important for statistical inference, especially in models where the parameters have a physical or biological interpretation.  For example, it is well-known that identifiability is a necessary condition for statistical consistency of maximum likelihood estimation \cite[Chapter 16]{Felsenstein2004}.

In phylogenetics, parameters of interest include the discrete tree parameter and numerical parameters specifying substitution processes on the edges of the tree.  For the simplest phylogenetic models, identifiability of both tree and numerical parameters have long been established \cite{Chang}.  But as models grow in complexity, with both the combinatorial description of trees and the underlying number of numerical parameters increasing, the question of identifiability is far from settled. 

A particular class of complex phylogenetic models of growing interest and use are the phylogenetic mixture models.  Relatively simple examples are the models with small numbers of parameters --- including those with $\Gamma$-distributed rates, invariable sites, and combinations of these --- that are currently the most commonly used in data analysis. More elaborate mixtures allow across-site rate variation with more freedom in the distribution of the rate multipliers \cite{HuelSuch}, the use of different rate matrices \cite{PagelMeade}, or even
multiple distinct trees each with their own rate and time parameters. Such models may have a large number of mixture components. For instance, 
a Bayesian nonparametric analysis conducted in \cite{HuelSuch} allowed a variable number of components, with a Dirichlet process prior specifying a mean of as many as 20.

However, only the simplest phylogenetic mixture models have been proven to be identifiable, typically where the number of parameters are small.   The papers \cite{AllmanAneRhodes07,APRS,Allman2006,Allman2008,Chai2010,StefVig2007}  contain previous results on identifiability of such models, of various sorts.  Note that only recently has it been shown that most choices of parameters of the widely-used GTR + I + $\Gamma$ model are identifiable \cite{Chai2010}, although for a certain type of rate matrix the question remains open.

\smallskip

Our goal in this paper is to develop methods to prove identifiability in phylogenetic models that are considerably more complex than in previous work.  In particular, we investigate the identifiability of phylogenetic models with many mixing components.  A consequence of our methods is the following theorem:

\begin{thm}\label{thm:intro} For an $r$-component identical tree mixture of the general Markov model of character evolution with $\kappa$-state random variables on an $n$-leaf binary phylogenetic tree, both the tree parameter and the numerical parameters are generically identifiable if $r  < \kappa^{\lceil n/4 \rceil - 1} $.  \end{thm}

By an \emph{identical tree mixture model} we mean a mixture of probability distributions coming from the same topological phylogenetic tree.  More complicated mixture models might have each distribution arising from a different topological tree.
Theorem \ref{thm:intro} improves substantially over past identifiability results on identical tree phylogenetic mixture models.  Previously, it was only known that the tree parameter is identifiable, and then only in the case that $r < \kappa$ (work of Allman and the first author \cite{Allman2006}).  This new theorem quantifies the intuition that larger taxon sets should allow for identifiability of more complex models, and is an exponential improvement over previous results.

\smallskip

Our strategy of proof is to combine two techniques coming from the algebraic study of phylogenetic models.  First, we use the representation of probability distributions in a phylogenetic model as tensors with small tensor rank and employ a theorem of J. Kruskal to uniquely identify components of that tensor.  Second, we use phylogenetic invariants as tools to identify deeply embedded features of phylogenetic trees, and to ``untangle'' probability distributions that have been shuffled together by the tensor analysis.  While each technique by itself is only able to make a small advance on the identifiability problem, when combined they give dramatically stronger results.  Background on these general techniques appears in Section \ref{sec:tensors}, and the proofs of the main theorems are in Section \ref{sec:commonsubstructure}.

Our techniques actually extend to mixtures from different trees provided they all share a certain type of  common substructure.  It is in this generality that we prove our main results, Theorems \ref{thm:numparam} and \ref{thm:treeparam},  with Theorem \ref{thm:intro} arising as  a corollary.

The assumption of any common substructure in the trees is of course false in some biological situations modeled by mixtures. For instance, if the mixture is due to the coalescent process modeling incomplete lineage sorting on a species tree of populations,  then components will be present from all topological gene trees \cite{Degnan2005,Wakeley}.  However, one might also model lateral gene transfer at a number of (unknown) locations in a tree as a mixture, and for this the assumption of common substructure could be quite plausible.


\section{Preliminaries} \label{sec:prelim}

\subsection{Mixture models}

Consider the general Markov model of $\kappa$-state character evolution, $GM(\kappa)$, on $n$-taxon trees ({\it e.g.,} $\kappa = 4$ corresponding to DNA sequences).  We assume the taxa labeling the leaves are identified with $[n] = \{1,2,\ldots, n\}$.  Then for each rooted leaf-labeled tree $T$, there is a \emph{parametrization map} $\psi_T$ giving the joint distribution of states at the leaves of the tree $T$ as functions of continuous parameters, which specify the state distribution at the root and the transition probabilities on the edges. 
Let $S_T$ denote the continuous parameter space of $GM(\kappa)$ on $T$, which is a full dimensional subset of some $\rr^m$.  Then
$$
\psi_T : S_T \rightarrow \Delta^{\kappa^n -1},
$$
where $\Delta^{\ell-1} \subseteq [0,1]^{\ell}$ denotes the probability simplex comprised of non-negative real vectors summing to $1$.  The image of this map is the phylogenetic model $\calm_T \subseteq \Delta^{\kappa^n -1}$.

The associated $r$-component mixture model has the following parametrization:  For every $r$-tuple of  trees $\bfT = (T_1, T_2, \ldots, T_r )$ on the same taxa $[n]$, let $S_\bfT = S_{T_1} \times \cdots \times S_{T_r} \times \Delta^{r-1}$ and let
$$
\psi_\bfT : S_\bfT \rightarrow \Delta^{\kappa^n -1},
$$
be defined by
$$
\psi_\bfT(s_1, \ldots, s_r, \pi)  = \pi_1 \psi_{T_1}(s_1) + \cdots + \pi_r \psi_{T_r}(s_r).
$$
Thus $\pi$ is the vector of mixing parameters; each $\pi_i$ gives the proportion of i.i.d.~sites that evolve along tree $T_i$ with parameter vector $s_i$.  The $r$-component mixture model on $\bfT$ is the image of the map $\psi_\bfT$, and is denoted $$\calm_{\bfT}=\calm_{T_1} \ast \calm_{T_2} \ast \cdots \ast \calm_{T_r}.$$  
Clearly $\calm_{\bfT}$ depends only on the unordered multiset of the trees in $\bfT$.
In the case where  $T_i=T$ for all $i$, we call this an $r$-component \emph{identical tree} mixture model on $T$.

\smallskip

We focus on the mixture models built from the basic model $GM(\kappa)$ in this paper, as these are quite general \emph{algebraic models}, for which the maps $\psi_T$ are naturally defined by polynomial formulas.  Many models which are not polynomial (in particular, those built from the general time-reversible model) can be embedded in them.  The polynomial structure of algebraic models allows them to be studied using techniques from algebraic geometry.

\subsection{Identifiability of parameters}

For algebraic models, it is convenient to slightly weaken
the notion of identifiability of parameters to \emph{generic identifiability}. The
word ``generic'' is used to mean ``except on a proper algebraic
subvariety'' of the parameter space. (See section \ref{subsec:inv} for a formal definition of variety.) Although it is sometimes
possible to be explicit about this subvariety, we usually are
not, since the key point in interpretation is that a proper subvariety
is a closed set of Lebesgue measure 0 inside the larger set.
Thus regardless of the precise subvariety involved, ÔrandomlyÕ
chosen points are generic with probability 1.

On an unmixed $GM(\kappa)$ model on a single tree $T$, there are several well-understood issues with identifiability of parameters. 
First, at any internal node of the tree, in a phenomenon called \emph{label swapping}, one may permute the names of the state space of the corresponding hidden variable (permuting the columns or rows of the Markov matrices on edges  leading to or from the node) with no effect on the probability distribution. Second, while the standard parameterization of the $GM(\kappa)$ model on a tree $T$ requires specification of the root of $T$, for generic choices of parameters one can relocate the root (with an appropriate uniquely determined change to the parameters, up to label swapping) with no effect on the probability distribution.  Third, if any internal nodes of $T$ have degree 2, they may be suppressed and the Markov matrices on incident edges combined, with no effect on the probability distribution. Thus one generally assumes trees have no such nodes. For simplicity,  we do not always explicitly refer to these issues in our formal statements in this article. However, we will occasionally use the second fact to choose a convenient location for a root of a tree in our arguments.

That these are the only issues for parameter identifiability for the unmixed model is the content of the following theorem, which was essentially shown in \cite{Chang}. 
\begin{thm} \label{thm:basictree} For the $GM(\kappa)$ model on a single tree,
\begin{enumerate}   
\item  The unrooted tree parameter is generically identifiable, in the class of binary trees.
\item  For a fixed binary tree $T$, the numerical parameters of the $GM(\kappa)$ model on $T$ are generically identifiable, up to label swapping at internal nodes of the tree, and an arbitrary choice of a node as the root.
\end{enumerate}

\end{thm}

An additional issue for identifiability of $r$-tree mixtures
is component swapping: Interchanging the trees along with their
parameters, while permuting the mixing parameters in the same way,
has no effect on the resulting distribution. A useful notion
of identifiability must allow for this.

\begin{defn}
The tree parameters of the $r$-tree mixture are generically identifiable if  
for any binary trees $\bfT = (T_1, \ldots, T_r)$ on the same set of taxa, and generic choices of parameters $s_1, \ldots, s_r, \pi$, 
$$
\psi_\bfT(s_1, \ldots, s_r, \pi)  =  \psi_{\bfT'}(s_1', \ldots, s_r', \pi')
$$
implies that $\bfT = \sigma \cdot \bfT'$ for some $\sigma\in \fraks_r$, the symmetric group of  permutations.
\end{defn}

We also investigate identifiability of tree parameters when restricting to specific classes of $r$-tuples of trees.  For example, Theorem \ref{thm:intro} concerns identifiability of tree parameters among all sets $\bfT = \{T_1, \ldots, T_r\}$, where $T_1 = \cdots = T_r$.   Our main results, Theorems \ref{thm:numparam} and \ref{thm:treeparam} concern identifiability in the class of $r$-tuples of trees that all contain a specified deep common substructure, whose precise definition will be given in Section \ref{sec:commonsubstructure}.

\begin{defn}
The continuous parameters of an $r$-tree mixture on $\bfT$ are generically identifiable if for generic choices of $s_1, \ldots, s_r$ and $\pi$, 
$$
\psi_\bfT(s_1, \ldots, s_r, \pi)  =  \psi_{\bfT}(s_1', \ldots, s_r', \pi')
$$
implies that there is a permutation $\sigma \in \fraks_r$ such that $\sigma \cdot \bfT = \bfT$, $s'_i = s_{\sigma(i)}$, and $\pi'_i = \pi_{\sigma(i)}$ for $i = 1, \ldots, r$. 
\end{defn}
Note this definition only allows the swapping of continuous parameters $s_i,\pi_i$ with $s_j,\pi_j$ when $T_i = T_j$.

\smallskip
\subsection{Splits and tripartitions}
We will use the combinatorial notion of a split of the leaves of a tree associated to an edge in a binary tree, as well as the analog of this concept for a node of the tree.

\begin{defn}
A split of $[n]$ is a bipartition $A|B$ of $[n]$ with two nonempty elements.  A split is said to be compatible with a tree $T$ if it arises as the partition of leaves induced by an edge in some binary resolution of $T$.  

Similarly, a tripartition of $A|B|C$ of leaves is said to be compatible with $T$ if it arises as the tripartition induced by an interior vertex in some binary resolution of $T$.

A collection of trees is said to have a common split (or tripartition) if the split (or tripartition) is compatible with every tree in the collection.  
\end{defn}

A collection of trees has a common tripartition $A|B|C$ if, and only if, it also the three common splits $A|B \cup C$, $B | A \cup C$, and $C|A \cup B$. For a binary tree, these are the splits associated to the edges radiating from the vertex inducing the tripartition.   Note also that our definition of compatible splits differs from the standard definition (\emph{e.g.}, in \cite{SemSt}) in the case of trees with polytomies.  Our notion is more useful when studying geometric properties of phylogenetic models.  


\section{Tensors and Invariants}\label{sec:tensors}

The two main tools we use to prove our results are Kruskal's theorem on uniqueness of tensor decompositions and phylogenetic invariants.  In this section, we describe these tools. Both are connected to the notion of a flattening of the probability distribution arising from a phylogenetic model.

\subsection{Tensors and Unique Decomposition} By a \emph{tensor}, we mean simply an $n$-way rectangular array of numbers. A  2-way tensor is thus a matrix.

For $j = 1,2, 3,$ let $M_j$ be an $r \times \kappa_j$ matrix with  $i$th row $\bfm^j_i = (m^j_i(1), \ldots, 
m^j_i(\kappa_j)$.  Let $[M_1, M_2, M_3]$ denote the 3-way $\kappa_1 \times \kappa_2 \times \kappa_3$ tensor defined by
$$
[M_1, M_2, M_3]  =  \sum_{i = 1}^r  \bfm^1_i \otimes \bfm^2_i \otimes \bfm^3_i.
$$
In other words, $[M_1, M_2, M_3]$ is an $\kappa_1 \times \kappa_2 \times \kappa_3$ array whose $(u,v,w)$ entry is 
$$
[M_1, M_2, M_3]_{u,v,w}  =  \sum_{i = 1}^r  m^1_i(u) m^2_i(v) m^3_i(w).
$$
Every 3-way tensor can be expressed in this way, for sufficiently large $r$. A nonzero tensor of this form with $r=1$ is said to have tensor rank 1. More generally, the minimal $r$ such that  a 3-way tensor can be decomposed as such a sum is called its \emph{tensor rank}. A natural question is when
this expression is essentially unique.

Note there are two basic operations on the matrices $M_1, M_2, M_3$ which leave unchanged the tensor $[M_1, M_2, M_3]$:  one can simultaneously permute the rows of the three matrices $M_1, M_2, $ and $M_3$, or taking three numbers $a_1, a_2, a_3$ such that $a_1a_2a_3 = 1$, one can replace the $i$th rows $\bfm^j_i$ by $a_j \bfm^j_i$.  Kruskal's Theorem \cite{Kruskal1976, Kruskal1977} describes a situation where these operations lead to the only variants in a tensor decomposition.

Given an $r \times \kappa$ matrix $M$, its \emph{Kruskal rank}, denoted ${\rm rank}_K(M)$, is the largest value $k$ such that every subset of $k$ rows of $M$ is linearly independent.  Note that ${\rm rank}_K(M) \leq {\rm rank}(M)$.

\begin{thm}[\cite{Kruskal1976,Kruskal1977}]\label{thm:kruskal}
Let $I_j = {\rm rank}_K(M_j)$, where $M_j$ is $r \times \kappa_j$.  If
$$I_1 + I_2 + I_3 \geq 2r +2 $$
then $[M_1, M_2, M_3]$ uniquely determines $M_1, M_2, M_3$ up to simultaneous permutation and scaling of the rows.
\end{thm}

Kruskal's theorem has proven useful for proving identifiability results of numerical parameters for both phylogenetic models \cite{Allman2009a}  and for other statistical models with hidden variables \cite{AllmanMatiasRhodes2009,AllmanMatiasRhodes2010}.  We will show how to combine this with other algebraic techniques to also deduce identifiability of tree parameters.

\subsection{Flattenings}
While Kruskal's theorem concerns 3-way tensors, the tensors arising in phylogenetics are usually $n$-way $\kappa \times \cdots \times \kappa$ tensors, corresponding to the $n$ leaves of a phylogenetic tree.  We will make frequent use of flattenings of  $n$-way tensors to lower order tensors.    A flattening of a $n$-way tensor is simply a reorganization of that tensor as a $k$-way tensor, with $k<n$, of larger dimensions.  We take a $\kappa_1 \times \cdots \times \kappa_n$ tensor $M$, with typical entry $M(u_1, \ldots, u_n)$, and a partition $A_1 |A_2| \cdots | A_k$ of $[n]$, and we represent this as a 
$$
\prod_{a \in A_1} \kappa_a  \times \cdots \times \prod_{a \in A_k} \kappa_a
$$
tensor $\tilde{M}$.  The $(u_1, \ldots, u_n)$ entry of $M$ becomes
the $( (u_a)_{a \in A_1} ,  \ldots, (u_a)_{a \in A_k})$ entry of $\tilde{M}$. 
That is, the indices for the new tensor $\tilde{M}$ are vectors of indices from the tensor $M$.

Given a partition $A_1 | A_2|\cdots | A_k$ of $[n]$, we denote the corresponding flattening of $M$ by $\flat_{A_1 | A_2|\cdots | A_k}(M)$.

\subsection{Invariants, Phylogenetic and Otherwise}\label{subsec:inv}

We begin with a little background on algebraic geometry (see \cite{CLOS} for more detail).  Let $\rr[p_1, \ldots, p_m]$ be the set of all polynomials in the variables (or indeterminates) $p_1, p_2, \ldots, p_m$, with coefficients in the real numbers, $\rr$.  Algebraic geometry studies the zero sets of collections of polynomials.  That is, to a collection of polynomials $f_1, f_2, \ldots, f_k \in \rr[p_1, \ldots, p_m]$ we associate the \emph{variety}  
$$
V(f_1, \ldots, f_k) =  \left\{ \bfa \in \rr^m :  f_1(\bfa) = f_2(\bfa) = \cdots = f_k(\bfa) = 0 \right\}.
$$
The fact that these geometric sets arise from polynomials vanishing implies they have important structural features.

Varieties arise in studying statistical models through describing models implicitly, rather than parametrically.  For a fixed statistical model $\calm\subseteq \Delta^{m-1}$, an \emph{invariant} of $\calm$ is a polynomial $f \in \rr[p_1, \ldots, p_m]$ such that $f(\bfa) = 0$ for all $\bfa \in \calm$.  In the case where $\calm$ is a phylogenetic model, such a polynomial is  called a \emph{phylogenetic invariant}.

Our main use in this paper for phylogenetic invariants is their connection to generic identifiability, through the following basic proposition from algebraic geometry.

\begin{prop}\label{prop:intvar}
Let $V_0$ and $V_1$ be two irreducible algebraic varieties, such as those arising from parameterized statistical models. Suppose $f_0$ is an invariant for $V_0$, and there exists a point $p_1 \in V_1$ with $f_0(p_1) \neq 0$  Then  $V_1 \not\subseteq V_0$, and the variety $V_0\cap V_1$ is of lower dimension than $V_1$. That is, generic points on $V_1$ lie off of $V_0$.
\end{prop}

Among the most important and elementary phylogenetic invariants are the ones that arise from edge flattenings of tensors.  

\begin{defn} Let $A|B$ be a split compatible with the tree $T$. 
An \emph{edge invariant} for $T$ is a phylogenetic invariant that can be expressed as  a minor ({\it i.e.}, the determinant of a submatrix) of the matrix $\flat_{A|B}(P)$.
\end{defn}

As an indication of how edge invariants can be used to identify combinatorial information on the tree underlying a phylogenetic model, we recall the following theorem concerning models on a single tree.  While this statement is well-known in the phylogenetic invariants literature,  Theorem \ref{lem:edgemixture} of this article provides a more
general extension to mixture models.

\begin{thm}\label{thm:unmixededge}
Suppose that $T_0$ and $T_1$ are two $n$-leaf trees such that for $i = 0,1$,  $A_i|B_i$ is a split compatible with $T_i$ and incompatible with $T_{1-i}$, and let $\calm_i$ denote the $\kappa$-state general Markov model $GM(\kappa)$ on $T_i$.
Then the $(\kappa + 1)$-minors of $\flat_{A_i | B_i}(P)$ vanish on $\calm_{T_i}$ and do not vanish on $\calm_{T_{1-i}}$, and thus are edge invariants for the first model but not the second. 
In particular, edge invariants
can be used to generically identify the tree topology. 
\end{thm}

Edge invariants have been the phylogenetic invariants most interesting for tree identifiability in the past, and contain enough information to reconstruct the combinatorial type of a single tree in some situations. However, we need some more complicated invariants to get more information in the case of the phylogenetic mixture models considered here.  We describe these invariants, discovered in several different contexts \cite{Allman2003,Strassen},  in matrix form.

\begin{thm}\label{thm:strassen}
Let $P$ be a $\kappa\times\kappa\times \kappa$ tensor giving a distribution from the $GM(\kappa)$ model on a 3-leaf tree.  For $i = 1, \ldots, \kappa$,  let $P_{(i)}$ be the matrix slice $P_{(i)} = (P(i,u,v))_{u,v}$.  Then 
$$
P_{(i)} \left ( \adj P_{(j)} \right )  P_{(k)} -  P_{(k)} \left ( \adj P_{(j)} \right ) P_{(i)}  = 0.
$$
\end{thm}

Here $\adj A$ denotes the classical adjoint of $A$, which is given by polynomial expressions in the entries of $A$. In the case of nonsingular $A$, $\adj (A)  =  \det (A)  A^{-1}$.


\section{Identifiability of Mixture Models with Common Substructure}\label{sec:commonsubstructure}

In this section, we prove our main result, that both tree parameters and numerical parameters are generically identifiable in a phylogenetic mixture model provided we restrict to multisets $\bfT$ of trees that all share a certain substructure. More precisely, we require that all trees in $\bfT$ have two splits in common. The number of mixing components that can be identified via our techniques will depend on the sizes of the sets in these splits.  As a corollary, we deduce Theorem \ref{thm:intro}, after showing that if all trees are the same, there is a ``deep'' internal vertex with two of its incident edges giving the requisite splits.

Before proceeding to the statements and proofs of the main theorems, we prove three lemmas.

\begin{lemma}\label{lem:edgemixture}
(Edge invariants for tree mixtures)

Consider the $GM(\kappa)$ mixture model on $r$ trees $\bfT = (T_1, \ldots, T_r)$. Let $A|B$ be a bipartition of the taxa, with $ r < \min(\kappa^{\#A-1},\kappa^{\#B-1})$
\begin{enumerate}
\item  \label{lemit:1} If $A|B$ is compatible with all trees in $\bfT$, then all $(r\kappa+1)$-minors of $\flat_{A|B}(P)$ vanish for all distributions $P$ arising from the model.
\item  \label{lemit:2} If $A|B$ is not compatible with at least one tree in $\bfT$,  then for generic distributions $P$ arising from the model at least one
 $(r\kappa+1)$-minor of  $\flat_{A|B}(P)$ does not vanish.
 
\end{enumerate}
\end{lemma}

\begin{proof}
The claims concerning  (non)vanishing of minors are equivalent to claims that $\flat_{A|B}(P)$ has rank at most $r\kappa$ in case \eqref{lemit:1}, and generically has rank greater than $r\kappa$ in case \eqref{lemit:2}. Therefore we focus on investigating ranks of flattenings. 

\smallskip

If $A|B$ is compatible with all trees in $\bfT$, then, by passing to binary resolutions of the $T_i$, we may assume it is a split associated to edge $e_i=(a_i,b_i)$ in $T_i$. Then one sees that 
$$\flat_{A|B}(P)=M_A^T Q M_B.$$ Here $Q$ is the $r\kappa\times r\kappa$ block-diagonal matrix whose $i$th $\kappa\times\kappa$ block gives the joint probability distribution of
states for the random variables at $a_i$ and $b_i$, weighted by the component proportion $\pi_i$.
The matrices $M_A$,$M_ B$ are stochastic, of sizes $r\kappa\times \kappa^{\#A}$,  
 $r\kappa\times \kappa^{\#B}$, with entries in the $i$th block of $\kappa$ rows giving probabilities of states of variables in $A,B$ conditioned on states at $a_i$,$b_i$. This factorization implies the claimed bound on the rank.
 
 \smallskip

Suppose next that $A|B$ is not compatible with at least one of the trees in $\bfT$, say $T_1$.  
To show that $\flat_{A|B}(P)$ generically has rank greater than $r\kappa$, it is enough to give a single choice of parameters producing such a rank. Indeed,  this follows from Proposition \ref{prop:intvar}, applied to the model and the variety of matrices of rank at most $r\kappa$.

\smallskip

To simplify this choice, for each $T_i$ with $i>1$ choose all  Markov matrices for all internal edges of $T_i$ to be the identity, $I_\kappa$. Since $T_1$ is not compatible with $A|B$,  by  Theorem 3.8.6 of \cite{SemSt},  it has an edge  $e=(c,d)$, with associated split $C|D,$ such that all four sets
$A\cap C$, $A\cap D$, $B\cap C$, $B\cap D$ are nonempty. For all internal edges of $T_1$ except $e$, choose
Markov matrices to be $I_\kappa$ as well. Since the effect of an identity matrix on an edge is the same as contracting that edge,  with these choices we need henceforth argue only in the following special case:  for $i>1$, $T_i$ is a star tree with central node $a_i$, and $T_1$ has the form of two star trees, on $C$ and on $D$, that are joined at their central nodes by $e$.  

Now express the distribution $P=P_1+P'$ where $P_1$ is the mixture component from $T_1$, and $P'$ the sum of the components on the star trees $T_2=\cdots =T_r$. 
Then, one can write
 $$M_2:=\flat_{A|B}(P')=N_{A}^T R N_{B},$$
 with $R$ an $(r-1)\kappa\times (r-1)\kappa$ diagonal matrix giving the distribution of states at $a_i$ in components $2,\dots, r$ weighted by the $\pi_i$, and $N_{A}$, $N_{B}$ are stochastic matrices of sizes $(r-1) \kappa\times \kappa^{\#A}$,  
 $(r-1)\kappa\times \kappa^{\#B}$ with entries giving conditional probabilities of states of variables in $A$, $B$ conditioned on states/components at the $a_i$. By choosing positive root distributions at the nodes $a_i$, and positive $\pi_i$, we ensure $R$ will have positive diagonal entries, and hence have full rank.  Furthermore, the rows of $N_{A},N_{B}$ are formed from the tensor product of corresponding rows of the Markov matrices on the edges of the star trees, and are thus generalized Vandermonde matrices.
(Recall that if $f_{1}, \ldots, f_{t}$ are a linearly independent set of polynomials, and $u_{1}, \ldots, u_{s}$ are points, the generalized Vandermonde matrix is the matrix $s\times t$ matrix with $i,j$ entry $f_{j}(u_{i})$. Here the polynomials $f_j$ are determined by the formulae for the entries in the tensor product of the rows, and the $u_i$ by the entries in the Markov matrices.)
A generalized Vandermonde matrix has full rank for generic choices of $u_{1}, \ldots, u_{s}$. Since $(r-1)\kappa<\min(\kappa^{\#A},\kappa^{\#B})$, for generic parameters $M_2$ has rank $(r-1)\kappa$.

On the other hand, consider $P_1$, where we choose all matrices on pendant edges of $T_1$ to be $I_\kappa$, and both the root distribution at $c$ and $M_e$ to have all positive entries.
Then
$$M_1:=\flat_{A|B}(P_1)=N_{1,A}^T R_1 N_{1,B},$$
where $R_1$ is a $\kappa^2\times \kappa^2$ diagonal matrix with entries giving the joint distribution at $c$ and $d$ weighted by $\pi_1$,
and $N_{1,A},N_{1,B}$ have all zero entries except for a single 1 in each row,  and full row rank. Thus $M_1$ has rank $\kappa^2$. Moreover, it has at most one non-zero entry in each row and column, so both $\im(M_{1})$ and $\ker(M_{1})$ are coordinate subspaces.

Since $\flat_{A|B}(P) = M_{1} + M_{2}$, our goal is to show that ${\rm rank}(M_{1} + M_{2}) > r \kappa$ for generic choices of the parameters not yet specified (the Markov matrices on the trees $T_2,\dots, T_r$). Without loss of generality assume that $\#A \geq \#B$, so to do this it is enough to make
\begin{equation} \label{eq:ranks}
{\rm rank}(M_{1} + M_{2}) = \min( (r-1) \kappa + \kappa^{2}, \kappa^{\#B} ). 
\end{equation}

We use the following facts about matrices:  
Let $X$ and $Y$ be $s \times t$ matrices.   With $\im(X),\ker(X)$ denoting the image and kernel of $X$ as a linear transformation from $\rr^{t}$ to $\rr^{s}$,
then $\im(X) \cap \im(Y) = 0$ implies $\ker(X+Y)= \ker X \cap \ker Y$.   
Also, if $\nullity(X + Y) = \nullity(X) + \nullity(Y) - t$, then by the rank/nullity theorem $\rank(X + Y) = \rank(X) + \rank(Y)$.  

\smallskip

First consider the case where $(r-1) \kappa + \kappa^{2} \leq \kappa^{\#B}$.  By the preceding paragraph, to show equation \eqref{eq:ranks} it suffices to choose parameters so that $\im(M_{1}) \cap \im(M_{2}) = 0$ and $\dim( \ker (M_{1}) \cap \ker (M_{2}) ) =  \nullity(M_{1}) + \nullity(M_{2}) - \kappa^{\#B}$.

Since generically $N_{A}$ and $N_{B}$ have full rank, it follows that $\im ( M_{2}) = \im ( N_{A}^T)$ and $ \ker(  M_{2})  = \ker (N_{B})$.  
But $\im (M_{1})$ is a coordinate subspace, so it intersects $\im (N_{A}^T)$ nontrivially if and only the submatrix of $N_{A}^T$ obtained by deleting rows corresponding to those coordinates has nontrivial kernel.  That submatrix is a $(\kappa^{\#A} - \kappa^{2})\times (r- 1)\kappa$ generalized Vandermonde matrix with $\kappa^{\#A} - \kappa^{2} \geq  \kappa^{\#B} - \kappa^{2} \geq (r- 1)\kappa$, so it has full column rank.  This proves that $\im ( M_{1}) \cap \im (M_{2}) = 0$ generically.

Since $\ker (M_{1})$ is also a coordinate subspace,  its intersection with $\ker (M_{2}) = \ker (N_{B})$ is isomorphic to the kernel of the submatrix of $N_{B}$ obtained by deleting the columns corresponding to required zero entries in vectors in $\ker (M_{1})$.  Since this submatrix  is a  $(r -1) \kappa\times (\kappa^{\#B} - \kappa^{2})$ generalized Vandermonde matrix, the dimension of this kernel is 
$$\kappa^{\#B} - \kappa^{2} - (r-1)\kappa =  (\kappa^{\#B} - \kappa^{2}) + (\kappa^{\#B} - (r-1)\kappa) - \kappa^{\#B}.$$
Thus $\dim(\ker(M_1) \cap \ker( M_2)) = \nullity(M_1) + \nullity(M_2) - \kappa^{\#B},$ so $\rank(M_{1} + M_{2}) =  (r-1)\kappa + \kappa^{2}$.

\smallskip

In the case where $(r-1) \kappa + \kappa^{2} > \kappa^{\#B} $ the same arguments as above apply after modifying our choices so all but $\kappa^{\#B} -(r-1)\kappa$ of the entries of $R_{1}$ are zero.  Then we deduce that we can choose $M_{2}$ so that $\rank(M_{1} + M_{2}) =(r-1)\kappa +\kappa^{\#B} -(r-1)\kappa= \kappa^{\#B}$.  
\end{proof}

Picking any internal vertex of a binary tree, the induced tripartition of the leaf variables allows us to create 3 agglomerate variables. In this way, we can view a phylogenetic model as one to which we can apply Kruskal's theorem.
More specifically, consider a probability distribution $P$ in the  $GM(\kappa)$ mixture model on trees
$\mathbf T=(T_1, \ldots, T_r)$, where the $T_i$ share a common tripartition $A|B|C$ of the leaves, arising from the vertices $v_i$.   Suppose $P_i$ is the weighted mixture component from $T_i$ in $P$. Then from the parameters on $T_i$, one can
give $\kappa\times \kappa^{\#A}$, $\kappa\times \kappa^{\#B}$, $\kappa\times \kappa^{\#C }$ stochastic matrices $M_{i,A}$,  $M_{i,B}$, $M_{i,C}$ of conditional probabilities of states at the leaves in $A$, $B$, $C$, given the state at $v_i$. Letting $\widetilde M_{i,A}$ be the matrix obtained from $M_{i,A}$ by multiplying rows by the corresponding entry of the root distribution at $v_i$ and by the weight $\pi_i$, one checks that $$\flat_{A|B|C}(P_i)=[\widetilde M_{i,A},  M_{i,B}, M_{i,C}].$$
Let $M_A$ denote the $r\kappa\times \kappa^{\#A}$ matrix obtained by stacking the $M_{i,A}$, and $M_B,M_C$ similarly be matrices obtained by stacking the $M_{i,B},M_{i,C}$. Then
 $$\flat_{A|B|C}(P)=[M_{A},  M_{B}, M_{C}].$$

To apply Kruskal's theorem to this flattening, we must first show that the technical conditions on Kruskal rank of the matrices apply, at least generically. 

\begin{lemma}\label{lem:applykruskal}
Consider an $r$-fold $GM(\kappa)$ mixture model on trees $\bfT = (T_1, \ldots, T_r)$ with a common tripartition $A|B|C$ of the leaves.
Then 
$$\flat_{A|B|C}(P)=[M_{A},  M_{B}, M_{C}]$$
for some matrices $M_{A},  M_{B}, M_{C}$ with $r\kappa$ rows. Moreover, for generic choices of the numerical parameters these matrices all have
full Kruskal rank (i.e., Kruskal row rank equal to their smaller dimension).
\end{lemma}

\begin{proof} 
The first claim was established in the discussion preceding the lemma.

For the second, by similar reasoning as was used in Lemma \ref{lem:edgemixture},
it is enough to show one choice of parameters gives these matrices full Kruskal rank.
By choosing matrix parameters on all internal edges of every $T_i$ to be the identity matrix, we may essentially assume every $T_i$ is the star tree, rooted at central node $v_i$. Choosing positive root distributions at $v_i$,  and positive mixing parameters $\pi_i$, it then suffices to only consider one set of leaves, say $A$.  

Now, as in the discussion of $N_A$ in the proof of Lemma \ref{lem:edgemixture}, one sees that $M_A$ is a generalized Vandermonde matrix. Since all its submatrices are also generalized Vandermonde matrices, it generically has
full Kruskal rank.
\end{proof}

The next lemma allows us to tease apart distributions which arise from mixing together slices of distributions from different trees. After  we have applied Kruskal's Theorem via Lemma \ref{lem:applykruskal},  it will be used  to identify which rows of the matrices arise from the same mixture component of the model.

\begin{lemma}[No Shuffling Lemma]\label{lem:noshuf}
Let  $T$, $T_1, \ldots, T_r$ be trees with $n \geq 3$ leaves, or $n \geq 4$ leaves if $\kappa = 2$.  For $i=1,\ldots, r$, let $P_i$ be a generic probability distributions from the $GM(\kappa)$ model on the tree $T_i$, scaled by positive constants $\pi_i$.    
For a fixed choice of $j\in[n]$, let $A|B=\{j\}|([n]\smallsetminus\{j\})$ and form the flattenings $\flat_{A|B}(P_i)$. Form  a new matrix from any $\kappa$ rows from these flattenings (with repeats allowed), and define $Q$ so that $\flat_{A|B}(Q)$ is this matrix.  Then $Q$ does not satisfy all the phylogenetic invariants for $T$ unless the chosen rows come from a single $P_i$ and  $T$ is a refinement of $T_i$.
\end{lemma}

\begin{proof}
Note that the multiplication by the $\pi_i$ has no effect on whether the tensor satisfies non-trivial invariants, because the phylogenetic varieties for the $GM(\kappa)$ model are invariant under the action of the general linear group at any leaf \cite{ARgm}.
\smallskip

Consider first the case that $n = 3$, and $\kappa \geq 3$. Suppose $Q$  is constructed from rows which come from at least two different $P_i$.   Without loss of generality, we assume $j=1$, so that in the notation of Theorem \ref{thm:strassen}, the slices $Q_{(i)}$ contain the entries of $Q$ arising from a single row of the flattening. We will show that $Q$ does not satisfy the invariants of that theorem.  

For the time being, treat two of these slices $Q_{(1)}, Q_{(2)}$ as fixed, and the third slice $Q_{(3)}$, which we may assume does not come from the same $P_i$ as either $Q_{(1)}$ or $Q_{(2)}$, as a variable.
Generically,  the matrix equation
\begin{equation}\label{eq:3inv}
Q_{(1)} \left ( \adj Q_{(2)}\right ) Q_{(3)}  -  Q_{(3)} \left ( \adj Q_{(2)}\right ) Q_{(1)} = 0
\end{equation}
then gives nonzero, linear constraints on the entries of $Q_{(3)}$.

However, for an arbitrary matrix $Q_{(3)}$ with positive entries whose sum is less than 1, we can find  a $P_j$ that has $Q_{(3)}$ as any designated slice. This shows that there exist such slices not satisfying equation \eqref{eq:3inv}, and hence, by Proposition \ref{prop:intvar}, that the generic slice does not. 
  
\smallskip

When $\kappa = 2$ and $n = 3$, there are no non-trivial invariants for $GM(\kappa)$  (those of Theorem \ref{thm:strassen} are
identically zero), hence we consider $n = 4$, and use the edge invariants of Theorem \ref{thm:unmixededge}.
But for any choice of 4-leaf tree, and choice of  index $j\in\{1, 2\}$, we can find a $P_i$ in the tree model so that $\flat_{A|B}(P)$ has any desired generic vector as its $j$th row.
Now  $Q$ is built from two such rows.  If the $P_1$ and $P_2$ that we take these slices from are not the same, then generically, we can choose those slices  to be arbitrary vectors. But then the flattening of $Q$ with respect to the split of $T$ will generically be a rank 4 matrix, and hence $Q$ will not satisfy the invariants for tree $T$.

\smallskip

For larger $n$, the result follows from the cases above by marginalization to 3- or 4-leaf trees.
\end{proof}

First we prove a theorem on the generic identifiability of numerical parameters in trees with a known common tripartition.

\begin{thm}\label{thm:numerical}
Suppose the trees $\bfT = (T_1, \ldots, T_r)$  have a known common  tripartition $A|B|C$, with $\#A \geq \#B \geq \#C$, and  $r \leq \kappa^{\#B -1}$.  If $\kappa=2$ also suppose $\#A\ge 3$. Then both $\bfT$ and the numerical parameters of the $GM(\kappa)$  mixture model on $\bfT$ are generically identifiable.
\end{thm}

\begin{proof}
Since the trees in $\bfT$ share a common tripartition $A|B|C$,  by Lemma \ref{lem:applykruskal} if a distribution $P$ arises from generic parameters of the model then $$\flat_{A|B|C}(P)=[M_A, M_B, M_C],$$
where $M_A$, $M_B$, and $M_C$  all have full Kruskal row rank, which will be $\min(r \kappa, \kappa^{\#A})$, $\min(r \kappa, \kappa^{\#B})$, and $\min(r \kappa, \kappa^{\#C})$, respectively.  According to Theorem \ref{thm:kruskal}, these matrices are uniquely determined up to simultaneous permutation and scaling of the rows provided
\begin{equation}\label{eq:n4kruskal}
\min(r\kappa, \kappa^{\#A}) + \min(r \kappa, \kappa^{\#B}) + \min(r \kappa, \kappa^{\#C}) \geq 2 r \kappa + 2.
\end{equation}
Since $\kappa \geq 2$ and $\#C \geq 1$, this inequality holds for all $r \leq \kappa^{\#B - 1}$.

At this point, we have recovered the matrices $M_A, M_B,$ and $M_C$ up to scaling and permuting the rows.  Each of the rows of the recovered $M_A$ will have entries from a scaled slice from a tree distribution on a subtree of one of the $T_i$ (the subtree spanning the vertex $v_i$ and all the leaves $A$).   We need to group these rows by the mixture components they come from.  However, the No Shuffling Lemma \ref{lem:noshuf}, says that generically it is possible to do this. Since ordering the rows of $M_A$ determines an order of the rows of $M_B,M_C$, we can then reassemble the flattened mixture components $P_i$ as the product $[M_{i,A}, M_{i,B},M_{i,C}]$ of appropriate submatrices $M_{i,A}, M_{i,B},M_{i,C}$ of $M_A,M_B,M_C$.

From $P_i$, we recover the mixing weight $\pi_i$ via
$$
\pi_i =  \sum_{(j_1, \ldots, j_n)  \in [\kappa]^n} P_i(j_1, \ldots, j_n).
$$
Then, by Theorem \ref{thm:basictree}, the tree $T_i$ and the numerical parameters  on it can be identified from $P_i/\pi_i$
\end{proof}

Now we proceed to prove identifiability of the numerical parameters and tree parameters  in our most general class of $r$-tree mixture models, the $j$-deep class.

\begin{defn}
For a positive integer $j$, the \emph{$j$-deep class} of $r$-tuples of trees $\bfT$ consists of all $r$-tuples of binary trees such that there exists a tripartition $A|B|C$ with $\#A,  \#B  \geq j$, $\#C \geq 1$, such that the splits $A| B \cup C$ and $ A \cup C|B$ are compatible with all trees in $\bfT$.
\end{defn}

Note that this definition does not require that $C | A \cup B$ be compatible with any of the trees in $\bfT$, so the full tripartition need not be associated to vertices in the $T_i$. The trees must only share two splits, each sufficiently deep in the tree. Furthermore, if $\bfT$ is in the $j$-deep class, we do not assume the tripartition is known, only that it exists.

\smallskip

We now prove our main theorems on identifiability of  parameters in $r$-tree mixtures. We state two versions, one for when a $j$-deep tripartition is known (including the case of when all the trees are known), and one for when it is not. The second of these requires a slightly stronger hypothesis on the number of mixture components.

\begin{thm}\label{thm:numparam}
Suppose $\bfT$ is in the $j$-deep class via a known tripartition $A|B|C$.
Then both $\bfT$ and the numerical parameters of the $GM(\kappa)$ mixture model associated to $\bfT$ are generically identifiable provided $r \leq \kappa^{j-1}$ and either $\kappa>2$, or $\kappa=2$ and $\#A\ge 3$.
\end{thm}

\begin{proof}
Fix some $c \in C$, let $D(c) = A \cup B \cup \{c\}$, and let $P_c =  P|_{D(c)}$ be the marginalization of $P$ to the leaves in $D(c)$.  This is a probability tensor for the mixture of induced trees $\bfT|_{D(c)}$, with numerical parameters obtained by restricting to these induced trees.  Note that the trees in $\bfT|_{D(c)}$ share the common tripartition $A | B | \{c\}$.  Thus Theorem \ref{thm:numerical} applies to identify
the trees $\bfT|_{D(c)}$ and numerical parameters on them. Then by Lemma \ref{lem:applykruskal} we may write
$$\flat_{A|B|\{c\}}(P_c)=[M_A, M_B, M_c],$$ and for generic choices of the numerical parameters, these matrices all have full Kruskal row rank. We may further specify that the rows of these matrices, in particular $M_A$, have been ordered into $r$ blocks of $\kappa$ rows, corresponding to the various mixture components. 

Note that since the matrix $M_A$ has full Kruskal row rank and is $r\kappa\times \kappa^{\#A}$ with $r\kappa\le \kappa^{\#A}$, it has full row rank.
Thus we may compute a left inverse $Q_A$, with $M_AQ_A=I_{r\kappa},$ the $r\kappa\times r\kappa$ identity.

Returning to the consideration of the full distribution $P$ and trees $\bfT$, we use $Q_A$ to disentangle the mixture components.
In each $T_i$ let $w_i$ be the node in the subtree spanning $A$ through which this subtree is connected to all other leaves. Then
$$\flat_{B\cup C|A}(P) =  M_{B\cup C}^T\Pi\widetilde M_{A},$$
where $\widetilde M_{A}, M_{B\cup C}$ are stochastic matrices of probabilities of states at the leaves in $A$, $B\cup C$ conditioned on  components and states at the $w_i$, and $\Pi$ is a diagonal matrix with entries the product of the mixing weights, $\pi_i$, and the root distributions at $w_i$. While the ordering of the mixture components  and root states in these matrices is arbitrary, we may assume it is the same as in the rows of $M_A$.
Then
$$M_A=R\widetilde M_A,$$
where and $R$ is a block diagonal matrix whose $i$th block gives conditional probabilities of state changes from $v_i$ to $w_i$ on $T_i$, and is generically invertible. 

Thus 
$$\flat_{B\cup C|A}(P) Q_A=M_{B\cup C}^T\Pi R^{-1} M_{A} Q_A=M_{B\cup C}^T\Pi R^{-1}.$$
This shows that by taking the columns of $\flat_{B\cup C|A}(P) Q_A$ in blocks of $\kappa$ we obtain entries associated to only one mixture component at a time. Moreover, multiplying a block of these columns  by the corresponding block of rows of $M_A=R\widetilde M_A$,
we obtain a flattened form of a single mixture component $\pi_iP_i$.

Summing the entries of $\pi_{i}P_i$ identifies $\pi_i$, and hence $P_i$. Then by Theorem \ref{thm:basictree} the tree $T_i$ and the numerical parameters on it are identifiable.
\end{proof}

\begin{thm}\label{thm:treeparam}
Suppose $\bfT$ is in the $j$-deep class.
Then both $\bfT$ and the numerical parameters of the $GM(\kappa)$ mixture model associated to $\bfT$ are generically identifiable provided $r < \kappa^{j-1}$.
\end{thm}

\begin{proof}
Since the  $\bfT$ is in the $j$-deep class and $r\kappa< \kappa^{\#A},\kappa^{\#B}$, for generic parameters we can use the edge invariants of Lemma \ref{lem:edgemixture} to find two splits $A|B \cup C$ and $B |A \cup C$ compatible with all trees in $\bfT$, with $\#A\geq \#B \geq j$, $\#C\geq 1$, simply by testing for all splits of an appropriate size.  

If $\kappa=2$,  then $2\le r <\kappa^{j-1}$ implies $j\ge 3$, so $\#A\ge 3$. Thus for any $\kappa\ge 2$, Theorem \ref{thm:numparam} applies to give the conclusion.
\end{proof}

We are now in a position to deduce Theorem \ref{thm:intro}, which will follow from Theorem \ref{thm:treeparam} and the following lemma.

\begin{lemma}\label{lem:specialvertex}
Let $T$ be an unrooted binary tree with $n \geq 3$ leaves.  Then there exists an internal vertex $v$ in $T$ inducing a tripartition $A|B|C$ such that two of the three components contain at least $\lceil n/4 \rceil$ leaves of $T$.
\end{lemma}

\begin{proof}
According to Exercise 1.5 in \cite{SemSt}, every tree has a centroid $v$, which is an internal node such that each component of $T \setminus v$ has at most $ |V|/2$ vertices where $V$ is the number of vertices of $T$.  This same statement holds if we replace $V$ with $n$ and vertices with leaves in the definition of the centroid.  Since the tree $T$ is binary  and $v$ is an internal vertex, there are three components of $T \setminus v$.  The largest component has at least $\lceil n/3 \rceil$ leaves and at most  $\lfloor n/2 \rfloor$.  Thus there are at least $\lceil n/2 \rceil$ leaves remaining between the other two components, which implies that in the most balanced case, one of the other two components has at least $\lceil n / 4 \rceil$ leaves.  Since $\lceil n/3 \rceil \geq \lceil n/4 \rceil$ this proves the claim.
\end{proof}

Simple examples show the bound  $\lceil n/4 \rceil$ in this lemma is the best possible.

\begin{proof}[Proof of Theorem \ref{thm:intro}]
According to Lemma \ref{lem:specialvertex}, there is an internal vertex of $T$ inducing a tripartition $A|B|C$ such that $\#A \geq \#B \geq \lceil n/4 \rceil$ and $\#C\ge 1$.  Thus $\bfT = (T, \ldots, T)$ is in the 
$\lceil n/4 \rceil$-deep class.  Theorem \ref{thm:treeparam} then applies.  
\end{proof}


\section{Further Directions}

The techniques employed in this paper have been primarily concerned with, and are effective for, the identification of parameters in mixture models where the underlying trees share large common substructures.  Establishing identifiability of either numerical or tree parameters in situations where there is no commonality between the trees remains an open problem.

Even in the case of general Markov mixtures of two 4-leaf trees little is understood: First, in the case of two different tree topologies being mixed, it is unknown if the tree parameters are generically identifiable. Second, if the two trees are given, it is unknown if numerical parameters are generically identifiable.  These problems might be addressed by finding stronger versions of the tensor rank results we have employed ({\it e.g.}, a strengthened version of Kruskal's theorem).  But it also seems likely that a solution to these these problems will require the development of new mathematical techniques.

\section*{Acknowledgement}
\label{sec:acknowledgement}
Thanks to John Huelsenbeck for stimulating this work through describing his own investigations with mixture models with many components. 

John Rhodes was partially supported by US National Science Foundation (DMS 0714830).
Seth Sullivant was partially supported by the David and Lucille Packard Foundation and the US National Science Foundation (DMS 0954865).

\bibliographystyle{plain}

\bibliography{tree}

\begin{thebibliography}{10}

\bibitem{AllmanAneRhodes07}
E.~S. Allman, C.~An\'e, and J.~A. Rhodes.
\newblock Identifiability of a {M}arkovian model of molecular evolution with
  gamma-distributed rates.
\newblock {\em Adv. in Appl. Probab.}, 40:229--249, 2008.
\newblock {\tt arXiv:0709.0531}.

\bibitem{AllmanMatiasRhodes2009}
E.~S. Allman, C.~Matias, and J.~A. Rhodes.
\newblock Identifiability of parameters in latent structure models with many
  observed variables.
\newblock {\em Ann. Statist.}, 37(6A):3099--3132, 2009.

\bibitem{AllmanMatiasRhodes2010}
E.~S. Allman, C.~Matias, and J.~A. Rhodes.
\newblock Parameter identifiability in a class of random graph mixture models,
  2010.

\bibitem{APRS}
E.~S. Allman, S.~Petrovic, J.~A. Rhodes, and S.~Sullivant.
\newblock Identifiability of two-tree mixtures for group-based models.
\newblock {\em IEEE/ACM Trans. Comput. Biol. Bioinformatics}, 2010.

\bibitem{Allman2003}
E.~S. Allman and J.~A. Rhodes.
\newblock Phylogenetic invariants for the general {M}arkov model of sequence
  mutation.
\newblock {\em Math. Biosci.}, 186(2):113--144, 2003.

\bibitem{Allman2006}
E.~S. Allman and J.~A. Rhodes.
\newblock {The identifiability of tree topology for phylogenetic models,
  including covarion and mixture models}.
\newblock {\em J. Comput. Biol.}, 13(5):1101--1113, 2006.

\bibitem{Allman2008}
E.~S. Allman and J.~A. Rhodes.
\newblock Identifying evolutionary trees and substitution parameters for the
  general {M}arkov model with invariable sites.
\newblock {\em Math. Biosci.}, 211(1):18--33, 2008.

\bibitem{ARgm}
E.~S. Allman and J.~A. Rhodes.
\newblock Phylogenetic ideals and varieties for the general {M}arkov model.
\newblock {\em Adv. in Appl. Math.}, 40(2), 2008.

\bibitem{Allman2009a}
E.~S. Allman and J.~A. Rhodes.
\newblock The identifiability of covarion models in phylogenetics.
\newblock {\em IEEE/ACM Trans. Comput. Biol. Bioinformatics}, 6(1):76--88,
  2009.

\bibitem{Chai2010}
J.~Chai and E.~A. Housworth.
\newblock {On Rogers's Proof of Identifiability for the {GTR} + {Gamma} + {I}
  Model}, 2010.
\newblock Preprint.

\bibitem{Chang}
J.~T. Chang.
\newblock Full reconstruction of {M}arkov models on evolutionary trees:
  identifiability and consistency.
\newblock {\em Math. Biosci.}, 137(1):51--73, 1996.

\bibitem{CLOS}
D.~Cox, J.~Little, and D.~O'Shea.
\newblock {\em {I}deals, {V}arieties, and {A}lgorithms: {A}n {I}ntroduction to
  {C}omputational {A}lgebraic {G}eometry and {C}ommutative {A}lgebra}.
\newblock Springer-Verlag, New York, second edition, 1997.

\bibitem{Degnan2005}
J.~H. Degnan and L.~A. Salter.
\newblock Gene tree distributions under the coalescent process.
\newblock {\em Evolution}, 59:24--37, 2005.

\bibitem{Felsenstein2004}
J.~Felsenstein.
\newblock {\em {Inferring Phylogenies}}.
\newblock Sinauer and Associates, 2004.

\bibitem{HuelSuch}
J.~P. Huelsenbeck and M.~A. Suchard.
\newblock A nonparametric method for accommodating and testing across-site rate
  variation.
\newblock {\em Syst. Biol}, 56(6):975--987, 2007.

\bibitem{Kruskal1976}
J.~B. Kruskal.
\newblock More factors than subjects, tests and treatments: an indeterminacy
  theorem for canonical decomposition and individual differences scaling.
\newblock {\em Psychometrika}, 41(3):281--293, 1976.

\bibitem{Kruskal1977}
J.~B. Kruskal.
\newblock Three-way arrays: rank and uniqueness of trilinear decompositions,
  with application to arithmetic complexity and statistics.
\newblock {\em Linear Algebra and Appl.}, 18(2):95--138, 1977.

\bibitem{PagelMeade}
M.~Pagel and A.~Meade.
\newblock Mixture models in phylogenetic inference.
\newblock In O.~Gascuel, editor, {\em Mathematics of Evolution and Phylogeny},
  pages 121--142. Oxford University Press, Oxford, 2005.

\bibitem{SemSt}
C.~Semple and M.~Steel.
\newblock {\em Phylogenetics}, volume~24 of {\em Oxford Lecture Series in
  Mathematics and its Applications}.
\newblock Oxford University Press, Oxford, 2003.

\bibitem{Strassen}
V.~Strassen.
\newblock Rank and optimal computation of generic tensors.
\newblock {\em Linear Algebra Appl.}, 52/53:645--685, 1983.

\bibitem{StefVig2007}
D.~\v{S}tefankovi\v{c} and E.~Vigoda.
\newblock Phylogeny of mixture models: Robustness of maximum likelihood and
  non-identifiable distributions.
\newblock {\em J. Comput. Biol.}, 14(2):156--189, 2007.

\bibitem{Wakeley}
J.~Wakeley.
\newblock {\em Coalescent {T}heory}.
\newblock Roberts and Company, 2008.

\end{thebibliography}

\end{document}